%% file: main.tex
\documentclass[12pt]{article}

\usepackage[T1]{fontenc}
\usepackage{amsthm,amsmath,amssymb,bm,bbm}

\usepackage{mathtools}
\usepackage{natbib}
\usepackage{multibib}
\usepackage[pdftex]{graphicx}
\usepackage{appendix}
\usepackage{enumitem}
\usepackage{caption}
\usepackage{subcaption}
\usepackage{url}
\usepackage{float}
\usepackage{tabularx}
\usepackage{booktabs}
\usepackage{algorithm}
\usepackage[noend]{algpseudocode}

\usepackage[usenames,dvipsnames,svgnames,table]{xcolor}
\definecolor{bestcolor}{gray}{0.5}
\usepackage[colorlinks=true,
            linkcolor=red,
            anchorcolor=blue,
            citecolor=MidnightBlue!70,
            urlcolor=black,
            hypertexnames=false]{hyperref}

\newcommand{\blind}{0}

\usepackage{setspace}
\usepackage{chngcntr}
\def\spacingset#1{\renewcommand{\baselinestretch}%
{#1}\small\normalsize} \spacingset{1}

\usepackage[labelformat=simple]{subcaption}

\counterwithin{equation}{section}

\newtheorem{thm}{Theorem}

\newtheorem{corollary}{Corollary}
\counterwithin{thm}{section}
\counterwithin{corollary}{section}

\counterwithin{rmk}{section}
\newtheorem{assump}{Assumption}

\allowdisplaybreaks
\def\bal#1\eal{\begin{align}#1\end{align}}
\def\balnn#1\ealnn{\begin{align*}#1\end{align*}}

\floatstyle{boxed}
\newfloat{sampler}{thp}{los}
\floatname{sampler}{Algorithm}

\newcites{Sup}{References}

\DeclareMathOperator*{\argmin}{arg\,min}

\DeclareMathOperator*{\diag}{diag}
\DeclareMathOperator*{\bigoh}{\mathcal{O}}
\DeclareMathOperator*{\spn}{span}
\DeclareMathOperator*{\tr}{tr}
\DeclareMathOperator*{\rank}{rank}

\DeclarePairedDelimiter\set{\{}{\}}
\DeclarePairedDelimiterX{\norm}[1]{\lVert}{\rVert}{#1}
\DeclarePairedDelimiterX{\abs}[1]{\lvert}{\rvert}{#1}

\newcommand{\activatebar}{%
    \begingroup\lccode`\~=`\|
    \lowercase{\endgroup\let~}\innermid
    \mathcode`|=\string"8000
}
\newcommand{\innermid}{\nonscript\;\delimsize\vert\nonscript\;}

\DeclarePairedDelimiterX{\expectarg}[1]{[}{]}{%
    \ifnum\currentgrouptype=16 \else\begingroup\fi
    \activatebar#1
    \ifnum\currentgrouptype=16 \else\endgroup\fi
}
\DeclarePairedDelimiterX{\variancearg}[1]{(}{)}{%
    \ifnum\currentgrouptype=16 \else\begingroup\fi
    \activatebar#1
    \ifnum\currentgrouptype=16 \else\endgroup\fi
}
\newcommand{\E}{\mathbb{E} \, \expectarg}
\newcommand{\Var}{\operatorname{Var}\variancearg}
\newcommand{\Cov}{\operatorname{Cov}\variancearg}
\newcommand{\Bias}{\operatorname{Bias}\variancearg}
\newcommand{\iidsim}{\overset{\text{iid}}\sim}

\def\mOne{{\mathbbm{1}}}

\newcommand{\bX}{\bm{\mathrm{X}}}
\newcommand{\bx}{\bm{\mathrm{x}}}

\newcommand{\bu}{\bm{\mathrm{u}}}
\newcommand{\be}{\bm{\mathrm{e}}}
\newcommand{\bmu}{\bm{\mu}}
\newcommand{\Sxy}{S_{Y \mid \bX}}
\newcommand{\SxAy}{S_{Y \mid \bX, \bX \in A}}

\newcommand{\Reals}[1]{\mathbb{R}^{#1}}
\newcommand{\ind}[1]{\mOne\{#1\}}

\newcommand{\dataset}{\mathcal{D}_n}
\newcommand{\LeftChild}{A_{L}}
\newcommand{\RightChild}{A_{R}}
\newcommand{\mtry}{m_{try}}

\begin{document}

\if0\blind
{
\title{ \bf Dimension Reduction Forests: Local Variable Importance using Structured Random Forests }
\author{
Joshua Daniel Loyal\hspace{.2cm}\\
Department of Statistics, University of Illinois at Urbana-Champaign\\
and \\
Ruoqing Zhu\hspace{.2cm}\\
Department of Statistics, University of Illinois at Urbana-Champaign\\
and \\
Yifan Cui\hspace{.2cm}\\
Department of Statistics, University of Pennsylvania \\
and \\
Xin Zhang\hspace{.2cm}\\
Department of Statistics, Florida State University
}
\date{}
\maketitle
}\fi

\if1\blind
{
  \bigskip
  \bigskip
  \bigskip
  \begin{center}
    \spacingset{1.5}
    {\LARGE\bf Dimension Reduction Forests: Local Variable Importance using Structured Random Forests}
  \end{center}
  \medskip
} \fi

\bigskip
\begin{abstract}
Random forests are one of the most popular machine learning methods due to their accuracy and variable importance assessment. However, random forests only provide variable importance in a global sense. There is an increasing need for such assessments at a local level, motivated by applications in personalized medicine, policy-making, and bioinformatics. We propose a new nonparametric estimator that pairs the flexible random forest kernel with local sufficient dimension reduction to adapt to a regression function's local structure. This allows us to estimate a meaningful directional local variable importance measure at each prediction point. We develop a computationally efficient fitting procedure and provide sufficient conditions for the recovery of the splitting directions. We demonstrate significant accuracy gains of our proposed estimator over competing methods on simulated and real regression problems. Finally, we apply the proposed method to seasonal particulate matter concentration data collected in Beijing, China, which yields meaningful local importance measures. The methods presented here are available in the {\tt drforest} Python package.

\end{abstract}

\noindent%
\vfill

\newpage
\spacingset{1.5}

\section{Introduction} \label{sec:intro}

Random forests~\citep{breiman2001} have repeatedly proven themselves as an effective supervised learning method. Random forests' competitive predictive accuracy, even when the problem is non-linear, high-dimensional, or involves complex interaction effects has resulted in their wide-scale application. Also, unlike many black-box methods, random forests provide an interpretable global variable importance measure for each predictor variable, which is of a broad interest in scientific problems in bioinformatics~\citep{diaz2006}, ecology~\citep{prasad2006}, and personalized medicine~\citep{laber2015}.

The success of random forests has sparked a desire to improve upon the original algorithm's shortcomings. A growing effort attempts to understand the statistical properties of random forests~\citep{lin2006, biau2012anlrf, mentch2016, wager2019grf, cui2019} such as the consistency and asymptotic normality of their predictions. Also, the core random forest methodology has been extended to numerous frameworks, including quantile regression~\citep{meinshausen2006}, survival analysis~\citep{hothorn2004, ishwaran2008, cui2019}, and causal inference~\citep{wager2018earlyci}. Furthermore, the tree-growing process has been modified to allow for additional randomness~\citep{geurts2006et}, linear combination splits~\citep{breiman2001, menze2011}, and non-uniform feature selection~\citep{amaratunga2008, zhu2015rltrees}. Lastly, the out-of-bag global variable importance measure has been extended to better account for correlated features through conditional independence tests~\citep{strobl2008}. For a comprehensive overview of random forest research developments, we refer the reader to the review paper by \citet{biau2016}.

In this work, we develop {\it dimension reduction forests} (DRFs): a new method for nonparametric regression that also quantifies local variable importance by using methods from sufficient dimension reduction (SDR)~\citep{li1991sir, bing2018}. Fields such as personalized medicine, policy-making, and bioinformatics are increasingly relying on these local assessments. For example, the ability to tailor medical treatment to patients based on their unique genetic makeup can drastically decrease their mortality. Currently, random forests have limited value in these applications because they lack a natural measure of local variable importance. To overcome this deficiency, we take the perspective of random forests as adaptive kernel methods. We pair random forests with sufficient dimension reduction to estimate a nonparametric kernel that adapts to the regression function's local contours. We then leverage this adaptivity to estimate a type of local variable importance we call {\it local subspace variable importance}. The result is a powerful model-free predictive method that is more accurate than naively combining random forests with global SDR methods. Before formally introducing our approach, we motivate how its local adaptivity overcomes certain deficiencies in traditional random forests.

% Motivation
\section{Motivation}\label{sec:motivation}

Throughout this article, we consider a general regression problem in which a continuous response $Y \in \mathbb{R}$ is predicted as a function of $p$ covariates $\bX = (X_1, \dots, X_p)^{\rm T} \in \Reals{p}$.

\subsection{Adapting the Random Forest Kernel to Local Structure}

A random forest (RF) is an ensemble of $M$ randomized decision trees. The random forest kernel between two points $\bx_0, \bx_1 \in \Reals{p}$ is
\begin{equation}\label{eq:rf_kernel}
K_{\text{RF}}(\bx_0, \bx_1) = \frac{1}{M}\sum_{m=1}^M \sum_{u = 1}^{\ell_m} \ind{\bx_0 \in A_u^m} \ind{\bx_1 \in A_u^m},
\end{equation}
where $\ell_m$ is the number of leaf nodes in the $m$th tree and $\set{A_1^m, \dots, A_{\ell_m}^m}$ contains the $m$th tree's leaf nodes with each $A_u^m \subseteq \Reals{p}$~\citep{scornet2016kernel}. This kernel is equivalent to the empirical probability that $\bx_0$ and $\bx_1$ share a leaf node within the random forest. Our goal is to extract a local variable importance measure from this kernel; however, this measure is only meaningful if the kernel reflects the regression function's local structure. Currently, the random forest kernel does not exhibit such adaptivity, which makes it difficult to extract local variable importance measures.

To be more concrete, consider the following non-linear regression model of a univariate response on a bivariate predictor: We draw $\bX_1, \dots, \bX_n$ independently from the uniform distribution on $[-3, 3]^2$, with response
\begin{equation}\label{eq:intro_example}
Y = 20 \max\left\{e^{-2 (X_1 - X_2)^2},
               \ 2 e^{-0.5(X_1^2 + X_2^2)},
               \ e^{-(X_1 + X_2)^2}\right\} + \varepsilon, \quad \varepsilon \overset{\text{iid}}\sim N(0, 1),
\end{equation}
and our goal is to extract local information about the regression (mean) function $m(\bx_0) = \E{Y | \bX = \bx_0 }$. Figure \ref{fig:kernel_comp} displays the contour lines of this function, which depend on three subspaces throughout the function's domain: $\mathcal{S}_1 = \spn(\be_1 - \be_2)$ (bottom-left and top-right), $\mathcal{S}_2 = \spn(\be_1 + \be_2)$ (top-left and bottom-right), and $\mathcal{S}_3 = \spn(\set{\be_1, \be_2})$ (center), where $\be_1$ and $\be_2$ are the standard basis. An ideal kernel should reflect the local structure by extending its weights along either the orthogonal complement of the one-dimensional subspaces, $\mathcal{S}_1$ and $\mathcal{S}_2$, or the contour lines within the two-dimensional subspace, $\mathcal{S}_3$. The random forest kernel does not exhibit this behavior.

Figure \ref{fig:kernel_comp} displays the random forest kernel estimated on data from the regression function in Equation (\ref{eq:intro_example}). The RF kernel cannot reflect the local structure because the leaf nodes, $\set{A_u^m}_{u=1}^{\ell_m}$, form an axis-aligned partition of the input space. Indeed, the RF kernel exhibits an axis-aligned plus shape that does not exploit the local dimension reduction structure. As a first encouraging result, our proposed dimension reduction forest's induced kernel, displayed in Figure \ref{fig:kernel_comp}, warps around the contour lines, improving its performance and accurately reflecting the regression function's local structure. We develop the dimension reduction forest kernel in Section \ref{sec:method}.

\begin{figure}[btp]
    \centering
    \includegraphics[width=\textwidth]{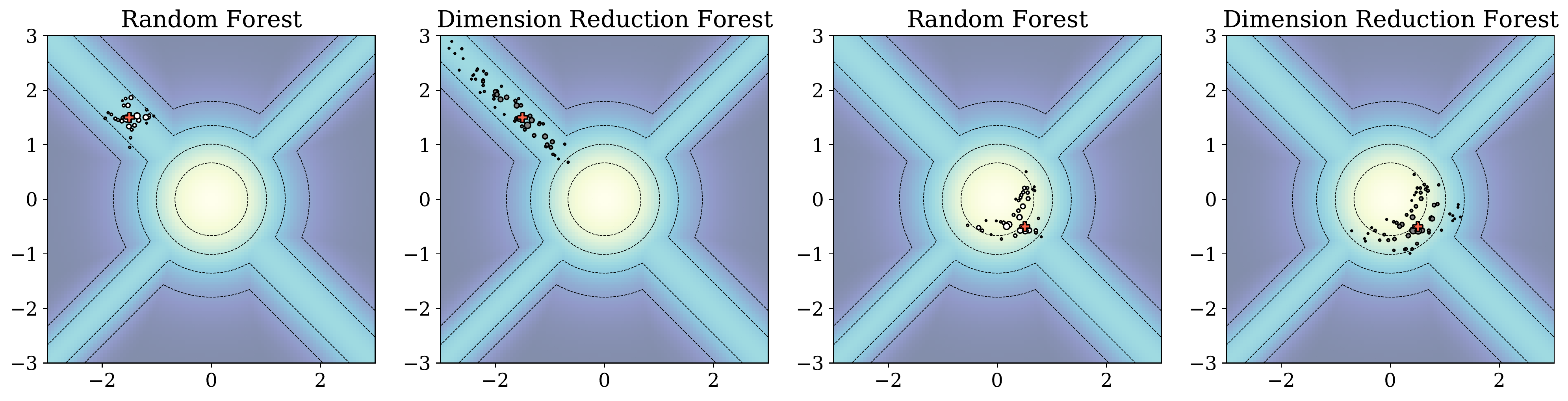}
    \caption{Kernels generated from a random forest and a dimension reduction forest on the problem defined in Equation (\ref{eq:intro_example}). The target prediction point $\bx_0$ is displayed with a red plus. The size of each circle is proportional to the kernel's weight. The contours of the true regression function are shown as black dashed curves.}
    \label{fig:kernel_comp}
\end{figure}

\subsection{Local Variable Importance}\label{subsec:vi}

Random forests' popularity is partly due to their ability to produce interpretable global variable importance measures; however, extracting local variable importance is an active research area. As previously mentioned, estimation and inference at the local level are becoming increasingly important, with personalized medicine applications such as inferring individualized therapeutic treatments. Our primary reason for constructing a kernel that can reflect local structure is to uncover such local insights.

As an illustration, we again consider the problem outlined in Equation (\ref{eq:intro_example}); however, now we draw $\bX_1, \dots, \bX_n$ independently from a $U[-3,3]^5$ so that there are three additional uninformative covariates. Random forests' global importance measure can indicate that $X_1$ and $X_2$ are informative; however, it cannot reveal the feature's heterogeneous contributions throughout the input space.

Figure \ref{fig:local_svi} displays a standard random forest's global permutation variable importance and our proposed local variable importance measure extracted from a dimension reduction forest at two query points. For regression, the importance assigned by a random forest to a feature is the decrease in mean squared error when the values of that feature are randomly permuted in the out-of-bag samples~\citep{ishwaran2018}. In contrast, DRF's local variable importance can be found in Algorithm~\ref{alg:lsi_algo} of Section~\ref{sec:method}. The random forest only identifies the two significant predictors, while the dimension reduction forest determines the variables' heterogeneous contribution. At $\bx_0 = (-1.5, 1.5, 0, 0, 0)^{\rm T}$, the dimension reduction forest correctly indicates that simultaneously increasing or decreasing both $X_1$ and $X_2$ influences changes in the regression function. In contrast, at $\bx_0 = (0.5, -0.5, 0, 0, 0)^{\rm T}$, the DRF indicates that simultaneously increasing $X_1$ while decreasing $X_2$, or vice versa, is influential. Unlike the random forest's global importance, the DRF's local importance results in two drastically different interventions.

\begin{figure}[btp]
\centering
\includegraphics[width=\textwidth]{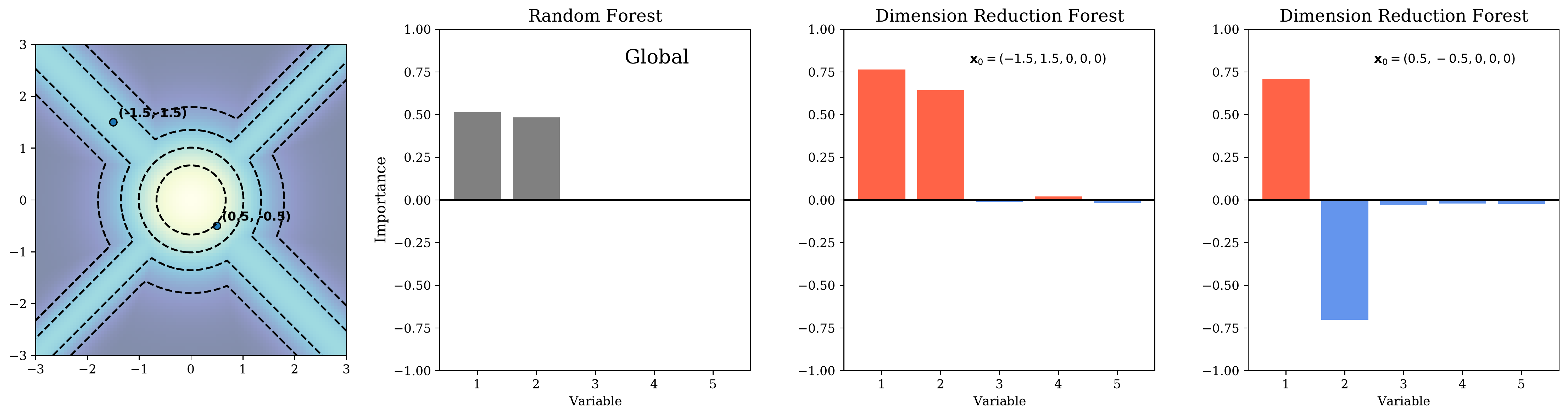}
\caption{Variable importance measures produced by a random forest and a dimension reduction forest.}
\label{fig:local_svi}
\end{figure}

\section{Dimension Reduction Forests}\label{sec:method}

Now, we formally introduce dimension reduction trees (DRTs) and dimension reduction forests (DRFs). We describe the DRF algorithm in Section \ref{subsec:splitting_rule}, connect the DRF splitting rule to an approximation of a locally adaptive kernel in Section \ref{subsec:adaptive_kernel}, and outline our method for extracting local variable importances from the induced DRF kernel in Section \ref{subsec:lsvi}.

\subsection{The Splitting Rule for Dimension Reduction Forests}\label{subsec:splitting_rule}

Dimension reduction forests are random forests composed of dimension reduction trees that use sufficient dimension reduction (SDR) techniques~\citep{li1991sir, bing2018} to approximate a locally adaptive kernel. Dimension reduction trees are built using a modified classification and regression tree (CART) algorithm \citep{breiman1984cart} with linear combination splitting rules. Specifically, a regression tree $T(\bx_0) = \sum_{u=1}^{\ell} \mu_{u} \ind{\bx_0 \in A_{u}}$, with $\ell$ leaf nodes $A_u \subseteq \Reals{p}$ and leaf means $\mu_u \in \Reals{}$, is recursively built using a training data set $\mathcal{D}_n = \set{(Y_1, \bX_i), \dots, (Y_n, \bX_n)}$ as follows. At any current leaf node $A$, we search for a splitting rule $\ind{\beta^{\rm T}_j\bx \leq c}$ for $\beta_j \in \Reals{p}$ and $j \in \set{1, \dots, d}$, for some $d \leq p$. We set the possible values of the threshold $c$ to all unique values of $\beta_j^{\rm T}\bX_i$ within node $A$. We then refine the current leaf node into two child nodes: $\LeftChild = A \cap \set{\beta^{\rm T}_j\bx \leq c}$ and $\RightChild = A \cap \set{\beta^{\rm T}_j \bx > c}$ with a certain splitting criterion. In the regression context, the best split minimizes the sum of squared errors in each child node:
\begin{equation}\label{eq:cart}
\argmin_{(\LeftChild, \RightChild) \in \mathcal{C}} \left\{ \sum_{\set{i \, : \, \bX_i \in \LeftChild}} (Y_i - \bar{Y}_{\LeftChild})^2 + \sum_{\set{i \, : \, \bX_i \in \RightChild}} (Y_i - \bar{Y}_{\RightChild})^2 \right\},
\end{equation}
where $\mathcal{C}$ is the set of all partitions formed by adding a rule of the form $\ind{\beta^{\rm T}_j\bx \leq c}$, and $\bar{Y}_{\LeftChild}$ and $\bar{Y}_{\RightChild}$ are the mean value of $Y$ inside $\LeftChild$ and $\RightChild$, respectively.  When we restrict the directions to the standard basis vectors, $\beta_j = \mathbf{e}_j$ for $j \in \set{1, \dots, p}$, the algorithm reduces to the axis-aligned splits used in the original CART algorithm. Algorithm \ref{alg:drf} outlines the procedure for building a dimension reduction forest.

\begin{sampler}[htbp]
Given a training data set $\mathcal{D}_n = \set{(\bX_i, Y_i)}_{i=1}^n$, do the following:
\begin{enumerate}
\item Draw $M$ bootstrap samples each with sample size $n$ from $\mathcal{D}_n$.
\item For the $m$th bootstrap sample, where $m \in \set{1, \dots, M}$, fit one DRT model $\hat{T}_m$ using the following rules:
\begin{enumerate}
\item {\it Optional variable screening}: At an internal node $A$, select the top $\mtry$ features that minimize Equation (\ref{eq:cart}) based on their corresponding axis-aligned splitting rules. Otherwise, set $\mtry = p$.
\item Estimate the leading sufficient dimension reduction direction $\hat{\beta} \in \Reals{p}$ using SIR ($\hat{\beta}_{\text{SIR}}$) and SAVE ($\hat{\beta}_{\text{SAVE}}$) with the features selected in step (a) and the samples in $A$ as input.
\item Project all observations onto the reduced space by calculating $[\hat{\beta}_{\text{SIR}}, \hat{\beta}_{\text{SAVE}}]^{\rm T}\bX_i$ for each sample $i$ in $A$. Construct a univariate splitting rule $\ind{\hat{\beta}^{\rm T} \bx \leq c}$ by minimizing Equation (\ref{eq:cart}) over the two features in the reduced data set.
\item Recursively apply steps (a)--(c) on each child node until the node's sample size is smaller than a pre-specified value $n_{min}$.
\item For each leaf node, set $\mu_u = \bar{Y}_{A_u} = \abs{\set{i \, : \, \bX_i \in A_u}}^{-1}\sum_{\set{i \, : \, \bX_i \in A_u}} Y_i$.
\end{enumerate}
\item Average the $M$ trees to obtain the final model, $\hat{f}(\bx) = M^{-1} \sum_{m=1}^M \hat{T}_m(\bx)$, with associated kernel $K_{\text{DRF}}(\bx_0, \bx_1)$ defined in Equation (\ref{eq:rf_kernel}).
\end{enumerate}
\caption{Dimension reduction forests.}
\label{alg:drf}
\end{sampler}

What differentiates DRFs from other random forest variants are splitting directions, $\beta$, estimated through inverse regression, a type of sufficient dimension reduction. Briefly, sufficient dimension reduction aims to find a matrix $B \in \Reals{p \times d}$, where $d \leq p$, such that the conditional distribution of $Y$ given $\bX$ is the same as $Y$ given $B^{\rm T} \bX$. In regression, this implies that
\begin{equation} \label{eq:sdr_model}
    Y = g(B^{\rm T}\bX, \varepsilon) = g(\beta_1^{\rm T}\bX, \dots, \beta_d^{\rm T}\bX, \varepsilon)
\end{equation}
for an unspecified link function $g(\cdot)$ and random noise $\varepsilon$. This model is equivalent to assuming $Y \perp \bX \mid B^{\rm T}\bX$. The column space of $B$, denoted by $\spn(B)$, is called a dimension reduction (DR) subspace. The intersection of all DR subspaces, assuming itself is also a DR subspace, is called the central subspace and denoted by $\Sxy$.

Since sufficient dimension reduction does not constrain the form of $g(\cdot)$, it is a powerful tool for designing flexible linear combination splitting rules. This flexibility differentiates our approach from existing oblique forests~\citep{menze2011, rainforth2015}, which impose restrictive linear constraints on the regression surface. Note that DRTs do not necessarily impose any global dimension reduction, such as in Equation (\ref{eq:sdr_model}). They only assume meaningful local dimension reduction structures develop as one progressively focuses on more local regions throughout the tree building process, i.e., $Y \perp \bX \mid B_A^{\rm T}\bX$ for $\bX \in A \subseteq \Reals{p}$ where $B_A \in \Reals{p \times d}$. Furthermore, we posit that using DRTs within the RF kernel improves its local statistical efficiency by aligning the kernel's weights with the local Hessian. We elaborate on this point in the next section.

To estimate the sufficient dimension reduction directions, $\set{\beta_1, \dots, \beta_d}$, we use SIR~\citep{li1991sir} and SAVE~\citep{cook1991save} since the former is efficient for small sample sizes and the latter is exhaustive. Both methods solve a generalized eigenvalue problem to estimate $B$, which we briefly overview in Section \ref{subsec:sir_save_algo} of the Supplementary Materials. Note that DRTs select the leading eigenvector instead of searching over the full DR subspace. Such an exhaustive search would be redundant because the eigenvalues associated with each direction already measure their local importance. Furthermore, to avoid choosing between SIR and SAVE, we search over the leading directions estimated by both methods when determining a split. This results in a splitting rule that combines the advantages of both SIR and SAVE and performs better than just applying a single method in each node. At the same time, we show that a DRF's computational complexity remains comparable to traditional random forests in Section~\ref{comp_complexity} of the Supplementary Materials.

Three problems arise when using traditional SDR methods within dimension reduction forests. First, SIR (SAVE) require at least $p$ samples for inference; however, it is often advisable to set $n_{min} < p$. To accommodate this setting, the DRT searches over axis-aligned splits when the internal node's sample size is less than $p$. Although this may affect the forest's ability to adapt to local dimension reduction structures, we found the degradation negligible in practice. Second, the use of dense linear combination splits limits DRF's applicability in medium to high-dimensional settings. For this reason, we introduce an optional variable screening step in Algorithm \ref{alg:drf} that selects covariates based on the minimization of Equation (\ref{eq:cart}). Finally, both SIR and SAVE require strong assumptions on the covariates' joint distribution for consistent estimation. In Section \ref{sec:theory}, we demonstrate that these assumptions only need to hold at the tree's root node, extending our method's applicability to any setting valid for SIR or SAVE.

\subsection{A Connection to Locally Adaptive Kernel Regression}\label{subsec:adaptive_kernel}

A recent research initiative aims to improve random forest's performance by accounting for inefficiencies in the induced random forest kernel function when the regression function has specific properties. For example, local linear forests~\citep{friedberg2018} modify the RF kernel to better model smooth signals through a local linear correction. Here, we argue that dimension reduction forests modify the RF kernel to better model signals with local dimension reduction structure by aligning the induced kernel with the local Hessian. The following argument is heuristic and only serves an illustrative purpose; however, the empirical studies in Section~\ref{sec:simulation} support the conclusions.

We begin by reviewing the bias-variance trade-off in multivariate kernel regression.
Consider the multivariate kernel $K(\bX_i, \bx_0) = \phi\left(\sqrt{(\bX_i - \bx_0)^{\rm T} H_{\bx_0} (\bX_i - \bx_0)}\right)$, where $\phi(\cdot)$ is a univariate kernel, such as the standard normal density, and $H_{\bx_0}$ is a positive definite $p \times p$ bandwidth matrix that depends on the prediction point $\bx_0$. Under certain regularity conditions, \citet{ruppert1994} showed that the bias and variance of the Nadaraya-Watson estimator, $\hat{m}(\bx_0) = (\sum_{i=1}^n K(\bX_i, \bx_0))^{-1}\sum_{i=1}^n K(\bX_i, \bx_0) \, Y_i$, conditioned on $\bX_1, \dots, \bX_n$ are
\begin{align*}
\Bias{\hat{m}(\bx_0)} &= \mu_2(\phi) \frac{m'(\bx_0)^{\rm T} H_{\bx_0} H_{\bx_0}^{\rm T}f'_{\bX}(\bx_0)}{f_{\bX}(\bx_0)}  +  \frac{1}{2} \mu_2(\phi) \tr(H_{\bx_0} m''(\bx_0)) + o_p(\tr(H_{\bx_0})), \\
\Var{\hat{m}(\bx_0)} &= \frac{1}{n \det(H_{\bx_0})} \norm{\phi}_2^2 \frac{\sigma^2(\bx_0)}{f_{\bX}(\bx_0)}\{1 + o_p(1)\}, \label{eq:kernel_var}
\end{align*}
where $\mu_2(\phi) = \int u^2 \phi(u) du$, $\norm{\phi}_2^2 = \int \phi(u)^2 du$, $\sigma^2(\bx_0) = \Var{Y | \bX = \bx_0}$, and the density of the covariates is $f_{\bX}(\cdot)$. Also, $m'(\bx_0)$ and $m''(\bx_0)$ indicate the gradient and Hessian of $m(\cdot)$ evaluated at $\bx_0$, respectively. The bias's first term, also known as the design bias, is often eliminated by applying a local linear correction~\citep{fan1996}. This correction was accounted for in random forests by \citet{friedberg2018}. For this reason, we focus on minimizing the second term involving the product of $H_{\bx_0}$ and $m''(\bx_0)$ while also controlling the variance.

In the presence of local dimension reduction structure, $m(\bx_0) = m(B_{\bx_0}^{\rm T} \bx_0)$ for some $B_{\bx_0} \in \Reals{p \times d}$, where $d < p$. In this case,
\begin{equation*}
\spn(m''(\bx_0)) = \spn(m''(B^{\rm T}_{\bx_0} \bx_0)) = \spn(B_{\bx_0} m''(\mathbf{u}_0) B_{\bx_0}^{\rm T}) \subseteq \spn(B_{\bx_0}),
\end{equation*}
where $\bu_0 = B_{\bx_0}^{\rm T} \bx_0$ and the Hessian in the third expression is taken with respect to $\bu_0$. This observation implies that $\rank(m''(\bx_0)) = d_0 \leq d$, where we assume $d_0 \neq 0$ for simplicity. According to the bias-variance decomposition given above, one can decrease a kernel estimator's variance without incurring any additional bias by increasing the bandwidth along the $(p-d_0)$-dimensional orthogonal compliment of $\spn(m''(\bx_0))$, which we denote by $\spn(m''(\bx_0))^{\perp}$.  Indeed, this can be accomplished by choosing a specific eigendecomposition of $H_{\bx_0} = U_{\bx_0}\Lambda_{\bx_0}U_{\bx_0}^{\rm T}$, where $U_{\bx_0}^{\rm T}U_{\bx_0} = I_p$ and $\Lambda_{\bx_0} = \diag(\lambda_{\bx_0,1}, \dots, \lambda_{\bx_0, p})$. Ignoring the design bias, if one selects $\spn(U_{\bx_0}) = \spn(m''(\bx_0))$, the bias and variance of the estimator in the new basis are
\begin{align*}\label{eq:lowdim_bias}
\Bias{\hat{m}(\bx_0)} &\propto\tr\left\{
\begin{pmatrix}
\Lambda_{\bx_0}^{(1)} & 0 \\
0 & \Lambda_{\bx_0}^{(2)}
\end{pmatrix}
\begin{pmatrix}
\Gamma_{d_0 \times d_0} & 0 \\
0 & 0_{(p - d_0) \times (p - d_0)}
\end{pmatrix}
\right\} = \tr(\Lambda_{\bx_0}^{(1)} \Gamma_{d_0\times d_0}), \\
\Var{\hat{m}(\bx_0)} &\propto \det(\Lambda_{\bx_0}^{(1)})^{-1}\det(\Lambda_{\bx_0}^{(2)})^{-1},
\end{align*}
where $\propto$ denotes proportionality between two values, $\Gamma_{d_0 \times d_0}$ is a diagonal matrix of the non-zero eigenvalues of $m''(\bx_0)$ and $\Lambda_{\bx_0}^{(1)} \in \Reals{d_0 \times d_0}$ and $\Lambda_{\bx_0}^{(2)} \in \Reals{(p - d_0) \times (p - d_0)}$ are also diagonal matrices. The previous two expressions demonstrate that increasing $\Lambda_{\bx_0}^{(2)}$, the bandwidth along $\spn(m''(\bx_0))^{\perp}$, reduces the variance without altering the asymptotic bias.

Of course, estimating the appropriate rotation matrix, $U_{\bx_0}$, and bandwidth matrix, $\Lambda_{\bx_0}$, for each sample point $\bx_0$ is computationally infeasible. However, the dimension reduction forest kernel approximates this procedure by recursively using the leading SDR direction to estimate an appropriate rotation and the splitting criterion in Equation (\ref{eq:cart}) to estimate an optimal bandwidth along this direction. We use SDR to estimate the splitting direction because the span of the local Hessian lies in the central subspace regardless of the form of $m(\cdot)$. Indeed, $\spn(m''(\bx_0)) \subseteq \spn(B_{\bx_0}) \subseteq \Sxy$ by definition.

\subsection{Local Subspace Variable Importance}\label{subsec:lsvi}

We propose a heuristic approach to extracting local variable importance from DRFs. Figure~\ref{fig:kernel_comp} demonstrates that the DRF kernel concentrates its bandwidth around the regression function's contour lines. The previous section gives intuition for this phenomenon, where we argued that the DRF kernel potentially adapts to the subspace spanned by the local Hessian. Note that when the local gradient is non-zero, we also expect this subspace to align with the gradient because $\spn(m'(B_{\bx_0}^{\rm T}\bx)) = \spn(B_{\bx_0} m'(\bu)) \subseteq \spn(B_{\bx_0})$.  In general, given an estimated kernel function, $K_{\text{DRF}}(\bx_0, \cdot)$, the subspace of minimum variance about the query point $\bx_0$ potentially aligns with the leading eigenvector of $B_{\bx_0}$ where $\spn(B_{\bx_0})$ is a local DR subspace. Formally, we define this direction as the {\it local subspace variable importance} (LSVI) at $\bx_0$. We estimate the LSVI at $\bx_0$ using the smallest principal component of the covariates' covariance matrix after centering the covariates around $\bx_0$ and weighting them by the dimension reduction forest kernel. We present the full algorithm for estimating LSVIs in Algorithm \ref{alg:lsi_algo}. Only the span of each LSVI is identifiable, which we interpret as the one-dimensional subspace that most influences $m(\bx_0)$.

\begin{sampler}[htbp]
Given a dimension reduction forest $\hat{f}(\bx)$ estimated with training data set $\mathcal{D}_n = \set{(\bX_i, Y_i)}_{i=1}^n$ and a query point $\bx_0$, do the following:
\begin{enumerate}
\item Calculate the sample weights $w_i = K_{\text{DRF}}(\bx_0, \bX_i)$, $1 \leq i \leq n$, where the RF kernel is defined in Equation (\ref{eq:rf_kernel}).
\item Center the variables at the query point: $\tilde{\bX}_i = \bX_i - \bx_0$ for $1 \leq i \leq n$.
\item Calculate the weighted sample mean $\bmu = \left(\sum_{i=1}^n w_i\right)^{-1} \ \sum_{i=1}^n w_i \tilde{\bX}_i$.
\item Output the eigenvector of $\left(\sum_{i=1}^n w_i\right)^{-1} \ \sum_{i=1}^n w_i (\tilde{\bX}_i - \bmu) (\tilde{\bX}_i - \bmu)^{\rm T}$ with the smallest eigenvalue.
\end{enumerate}
\caption{Local subspace variable importance.}
\label{alg:lsi_algo}
\end{sampler}

% Theory
\section{Theoretical Results} \label{sec:theory}

Here, we demonstrate that the conditions necessary for SIR and SAVE to recover the central subspace are also sufficient for dimension reduction forest's local splitting mechanism to recover a local central subspace.  We provide the proofs in Section~\ref{subsec:sir_save_proofs} of the  Supplementary Materials.

First, we introduce some notation. Let $A \subseteq \mathbb{R}^p$ denote the subset of the domain occupied by an internal node $A$ in a DRT. We denote the local central subspace of data contained in $A$ by $S_{Y \mid \bX, \bX \in A}$. Also, let $B_A \in \Reals{p \times d}$ denote a matrix such that $\spn(B_A) = S_{Y \mid \bX, \bX \in A}$, i.e., $Y \perp \bX \mid B_A^{\rm T} \bX$ for data in $A$.

Known results from the SDR literature demonstrate that the following two assumptions are sufficient for SIR (SAVE) to recover a subspace of $\SxAy$ when applied to data in an internal node $A$.

\begin{assump}[Local Linearity Condition]\label{asmp:lcm}
$\E{\bX|B_A^{\rm T}\bX, \bX \in A}$ is a linear function of $B^{\rm T}_A\bX$.
\end{assump}

\begin{assump}[Local Constant Variance Condition]\label{asmp:ccv}
$\Var{\bX|B_A^{\rm T}\bX, \bX \in A}$ is a nonrandom matrix.
\end{assump}

\noindent When $A = \Reals{p}$, Assumption \ref{asmp:lcm} and Assumption \ref{asmp:ccv} reduce to the usual linearity and constant variance conditions commonly made in the SDR literature. However, the DRT algorithm also requires the recovery of directions in $\SxAy$ within $A$'s child nodes where the local linearity and constant variance conditions may not hold. Remarkably, the following results address this requirement by showing that as long as Assumptions \ref{asmp:lcm} and Assumption \ref{asmp:ccv} hold in the parent node $A$, SIR (SAVE) applied to data in $A$'s children still recovers directions within the local central subspace, $\SxAy$.

\begin{thm}\label{thm:sir}
Assume that model (\ref{eq:sdr_model}) is satisfied in the parent node $A$, i.e., $Y \perp \bX \mid B_A^{\rm T}\bX$ for some $B_A$ if $\bX \in A$. Further assume Assumption \ref{asmp:lcm} holds within this node. If we implement a splitting rule of the form $\ind{\beta^{\rm T} \bx \leq c}$, where $\beta \in \spn(B_A)$, then within each child node $A_L = A \cap \{\beta^T \bx \leq c\}$ and $A_R = A \setminus A_L$ we have that
\begin{equation*}
\Sigma_A^{-1} \E{\bX|Y, \bX \in A_L} \in \spn(B_A) \quad \text{ and } \quad \Sigma_A^{-1} \E{\bX|Y, \bX \in A_R} \in \spn(B_A).
\end{equation*}
\end{thm}

\begin{thm}\label{thm:save}
Assume that model (\ref{eq:sdr_model}) is satisfied in the parent node $A$, i.e., $Y \perp \bX \mid B_A^{\rm T}\bX$ for some $B_A$ if $\bX \in A$. Further assume  Assumption \ref{asmp:lcm} and Assumption \ref{asmp:ccv} hold within this node. If we implement a splitting rule of the form $\ind{\beta^{\rm T} \bx \leq c}$, where $\beta \in \spn(B_A)$, then within each child node $A_L = A \cap \{\beta^T \bx \leq c\}$ and $A_R = A \setminus A_L$ we have that
\begin{equation*}
\Sigma_A - \Var{\bX | Y, \bX \in A_L} \in \Sigma_A \spn(B_A) \quad \text{ and } \quad \Sigma_A - \Var{\bX | Y, \bX \in A_R} \in \Sigma_A \spn(B_A).
\end{equation*}
\end{thm}
\noindent These theorems combined with known results from the dimension reduction literature, Corollary 3.1 and Proposition 5.1 in \citet{bing2018}, imply the ability of SIR (SAVE) to recover directions within $\SxAy$ when applied in the child nodes.

Note that Theorem~\ref{thm:sir} and Theorem~\ref{thm:save} still hold when $A_L$ and $A_R$ are replaced by any ancestral node $A_C \subseteq A$ constructed by repeatedly applying splitting rules of the form $\ind{\beta^{\rm T}\bx \leq c}$, where each $\beta \in \spn(B_A)$. The consequence of this observation is important: When the linearity and constant variance conditions hold for the covariates' joint distribution at the tree's root node, where $A = \Reals{p}$, then the DRT's splitting directions $\beta$ are within the central subspace $\Sxy$. In other words, we do not need to repeatedly re-check the linearity and constant variance conditions at each internal node for valid inference.

\section{Simulation Studies} \label{sec:simulation}

We provide several simulation studies to evaluate the performance of dimension reduction forests. The studies aim to evaluate the DRF's predictive performance and analyze our method's accuracy in generating local subspace variable importances.

\subsection{Predictive Performance}\label{subsec:predict}

We compared DRFs with two major competitors: traditional random forests as implemented in the {\tt scikit-learn} Python package and Nadaraya-Watson kernel estimators (NW Kernel). We also included two variations on the traditional random forest (SIR + RF and SAVE + RF) and the Nadaraya-Watson kernel estimator (SIR + NW Kernel and SAVE + NW Kernel). As their names suggest, these two variants first extract global SDR predictors using SIR or SAVE, and then train a random forest or an Nadaraya-Watson kernel estimator using only these features. This approach of using global SDR predictors within a kernel estimator is common practice, see \citet{adragni2009}. We included these methods to demonstrate the benefit of the recursive dimension reduction performed by DRFs.

We evaluated predictive performance using 50 repeated train-test splits on seven known regression functions. Each split contained $n_{train} = 2000$ training samples and $n_{test} = 1000$ testing samples. We recorded each method's mean squared error (MSE) on the test set. The simulations were repeated 50 times with different random seeds. Since machine learning methods are sensitive to hyperparameter settings, it is crucial to explore a large space of values. Table~\ref{tab:param_tune} in the Supplementary Materials lists all hyperparameter settings used in our experiments. To compare each method fairly, we reported the lowest test error achieved by each method over all parameter settings. This procedure differs from the standard practice of using cross-validation to choose the hyperparameters before evaluating on an independent test set. Our approach mitigates the impact of parameter tuning by reporting each method's best performance over all parameter settings.

For the predictive task, we considered the following four regression functions:

\begin{itemize}[nosep]
\item {\it Simulation 1:} $Y = 20 \max\left\{e^{-2 (X_1 - X_2)^2},\ 2 e^{-0.5(X_1^2 + X_2^2)}, \ e^{- (X_1 + X_2)^2}\right\} + \varepsilon$,
where $\bX_i \iidsim U[-3, 3]^{5}$ and $\varepsilon_i \iidsim N(0, 1)$. We presented this regression function in the motivation.

\item {\it Simulation 2:} $Y = 20 \max\left\{e^{-18 X_1^2},\ e^{-18 X_2^2}, \ 1.75 e^{-20 (X_1 + X_2)^2}, \ 1.75 e^{-20 (X_1 - X_2)^2} \right\} + \varepsilon$,
where $\bX_i \iidsim U[-1, 1]^{5}$ and $\varepsilon_i \iidsim N(0, 1)$.

\item {\it Simulation 3:} $Y = (\bX^{\rm T} \beta_1)^2 + (\bX^{\rm T} \beta_2)^2 + 0.5 \varepsilon$. Let $\mathbf{0}_n$ denotes an $n$-dimensional vector of zeros. We generate $\bX_i \overset{\text{iid}}\sim N(\mathbf{0}_{12}, \Sigma)$ and set $\Sigma_{ij} = 0.5^{\abs{i-j}}$ resulting in  moderate correlation between the covariates. The coefficients are $\beta_1 = (1, 1, 1, 1, 1, 1, \mathbf{0}_6^{\rm T})^{\rm T} / \sqrt{6}$ and $\beta_2 = (1, -1, 1, -1, 1, -1, \mathbf{0}_6^{\rm T})^T / \sqrt{6}$. Lastly, $\varepsilon_i \iidsim N(0, 1)$.

\item {\it Simulation 4:} $Y = \bX^{\rm T} \beta_1 (\bX^{\rm T} \beta_2)^2 + (\bX^{\rm T} \beta_3)(\bX^{\rm T} \beta_4) + 0.5 \varepsilon$,
where $\bX_i \overset{\text{iid}}\sim N(\mathbf{0}_{10}, \Sigma)$, $\Sigma_{ij} = 0.5^{\abs{i-j}}$, $\beta_1 = (1, 2, 3, 4, \mathbf{0}_6^{\rm T})^{\rm T} / \sqrt{30}$, $\beta_2 = (-2, 1, -4, 3, 1, 2, \mathbf{0}_4^{\rm T})^T$ $/ \sqrt{35}$, $\beta_3 = (\mathbf{0}_4^{\rm T}, 2, -1, 2, 1, 2, 1)^T / \sqrt{15}$, and $\beta_4 = (\mathbf{0}_6^{\rm T}, -1, -1, 1, 1)^{\rm T} / 2$. Lastly, $\varepsilon_i \iidsim N(0, 1)$.
\end{itemize}
We also include the three Friedman functions first introduced in \citet{friedman1991}. We provide the Friedman functions' definitions in Section~\ref{simulation_details} of the Supplementary Materials.

We designed the first two simulations to test the DRF's ability to adapt to local dimension reduction structure. Simulation 1 is the motivating example in Section~\ref{sec:motivation}. Simulation 2 contains four different local one-dimensional subspaces; two are axis-aligned, and two lie in a rotated coordinate system. The next two simulations are commonly used to test SDR methods. We included these regression functions to compare the DRF with a standard random forest trained on SIR and SAVE input features. Finally, Friedman 1--3 are classic regression problems used to test the performance of random forest methods. Here, our goal is to compare DRFs with RFs on problems where DRFs have no clear advantage.

Table \ref{tab:sim_results_pred} displays each method's mean percentage (100\%) improvement over the RF baseline and the corresponding standard deviations for the seven simulation scenarios. The best performing method for each scenario is shaded in gray. A negative value indicates degraded performance, a value of zero indicates no improvement, and a positive value indicates an improvement. Overall, the proposed DRF model performs well in all simulations.

\begin{table}[tbp]
    \resizebox{\columnwidth}{!}{
    \begin{tabular}{@{}llllllll@{}} \toprule
    & DRF & SIR + RF & SAVE + RF& NW Kernel & SIR + NW Kernel & SAVE + NW Kernel \\
    \bottomrule[\lightrulewidth]
Simulation \#1 & \cellcolor{bestcolor!25} ${23.51\ (6.12)}$ & $-1180.23\ (221.62)$ & $-62.95\ (43.50)$ & $-1116.26\ (107.69)$ & $-1120.81\ (108.71)$ & $-1119.24\ (107.56)$ \\
Simulation \#2 & \cellcolor{bestcolor!25} ${63.89\ (6.11)}$ & $-552.69\ (120.45)$ & $8.55\ (25.28)$ & $-585.66\ (95.20)$ & $-587.05\ (95.87)$ & $-586.41\ (95.24)$ \\
Simulation \#3 & $69.76\ (5.29)$ & $-155.37\ (62.01)$ & \cellcolor{bestcolor!25} ${71.30\ (5.02)}$ & $-29.21\ (15.06)$ & $-197.03\ (40.96)$ & $-29.21\ (15.06)$ \\
Simulation \#4 & \cellcolor{bestcolor!25} ${28.11\ (7.11)}$ & $-41.96\ (18.88)$ & $6.86\ (14.09)$ & $-19.60\ (10.91)$ & $-31.81\ (12.63)$ & $-21.92\ (12.89)$ \\
Friedman \#1 & \cellcolor{bestcolor!25} ${25.73\ (3.29)}$ & $-66.33\ (11.77)$ & $1.43\ (10.80)$ & $-123.96\ (13.86)$ & $-124.39\ (13.89)$ & $-124.37\ (13.89)$ \\
Friedman \#2 & \cellcolor{bestcolor!25} ${35.67\ (10.22)}$ & $-613.46\ (178.10)$ & $-43.85\ (34.41)$ & $-2158.87\ (264.96)$ & $-2162.95\ (265.29)$ & $-2161.14\ (264.99)$ \\
Friedman \#3 & \cellcolor{bestcolor!25} ${0.22\ (1.07)}$ & $-1.33\ (1.51)$ & $-1.29\ (1.53)$ & $-0.83\ (1.27)$ & $-0.84\ (1.27)$ & $-0.85\ (1.27)$ \\
    \bottomrule
    \end{tabular}}
    \caption{Percentage (100\%) improvement over the RF baseline: 1 - (MSE of method) / (MSE of random forest). The standard deviation over the 50 runs is displayed in parentheses. Methods with the best improvement (if any) over the RF are shaded in gray.}
    \label{tab:sim_results_pred}
\end{table}

DRFs see some of their largest gains over the competitors in Simulation 1 and Simulation 2. We expected this behavior because these regression surfaces exhibit heterogeneity in the dependent local subspace. Our intuition is that the DRF's recursive estimation of the local subspace allows it to modify its splitting direction based on the prediction point. The fact that the SIR + RF and SAVE + RF baselines, which can only split along the global DR subspace, have a much higher MSE lends credibility to this assertion.

Simulation 3 and Simulation 4 demonstrate that the DRF's performance does not degrade when a global SDR subspace adequately describes the regression surface. In Simulation 3, the DRF performs similarly to the SAVE + RF method due to the presence of a global quadratic response surface. However, the DRF drastically outperforms all methods on Simulation 4, which contains a regression surface that is troublesome for SIR and on which SAVE sees only a modest improvement. We conclude that the DRF is a flexible alternative to trying different SDR methods as inputs to nonparametric regressors.

The last three simulations (Friedman 1--3) show that the DRF performs well even when there is no dimension reduction beyond sparsity. In this case, the DRF performs the best or equivalent to the RF baseline. These results demonstrate that the DRF is a reasonable drop-in replacement for the traditional random forest.

\subsection{Local Subspace Variable Importance}\label{subsec:lsvi_sim}

Next, we evaluated how the LSVI estimates' (Algorithm \ref{alg:lsi_algo}) accuracy scaled with (1) a decreasing signal-to-noise ratio and (2) an increase in the number of uninformative covariates.  We considered the following four regression functions. {\it Simulation 1:} $Y = \abs{X_1} + \abs{X_2} + \varepsilon$, {\it Simulation 2:} $Y = X_1 + X_2^2 + \varepsilon$, {\it Simulation 3:} $Y = 5 \max\left\{e^{-0.25 X_1^2}, e^{-0.25 X_2^2}\right\} + \varepsilon$, and {\it Simulation 4:} $Y = 20 \max\left\{e^{-2 (X_1 - X_2)^2}, \ 2 e^{-0.5(X_1^2 + X_2^2)}, \ e^{- (X_1 + X_2)^2}\right\} + \varepsilon$. For each simulation, we drew $\bX_i \iidsim U[-3, 3]^p$ and $\varepsilon_i \iidsim N(0, \sigma^2)$.

We chose these functions because they exhibit a variety of local dimension reduction structures. Simulation 1 contains the simplest structure with only four local directions that vary by quadrant in the $X_1$-$X_2$ plane. Simulation 2 contains a simple analytic gradient, $(1, X_2)^{\rm T}$, that varies along a single direction. This simulation tested whether the LSVI estimates capture smoothly varying local structures. Simulation 3 is composed of a cross of simple one-dimensional axis-aligned subspaces. Simulation 4 is the function in the motivating example, which includes a highly variable local structure.

To evaluate the LSVI estimates' accuracy, we generated $n = 2000$ samples from each regression function and selected $n_{test} = 100$ points uniformly at random to test the LSVIs. We recorded the trace correlation between the estimated LSVIs and the normalized gradients at the randomly chosen test points. We repeated each simulation 50 times with different random seeds. To evaluate the effect of a degrading signal-to-noise ratio, we performed the experiments with $p = 10$  covariates and set $\sigma^2$ so that the signal-to-noise ratio was 5:1, 3:1, 3:2, 1:1, and 3:4. We define the signal-to-noise ratio as $\Var{\E{Y | \bX}} \, / \, \sigma^2$.

We compared the LSVIs generated by the following four methods: DRF, SAVE, SIR, and Local SIR. LSVI estimation is sensitive to the DRF's minimum leaf node size. We searched over $n_{min} \in \set{3, 10, 25, 50, 100}$ at each prediction point and reported the best direction in terms of trace correlation to mitigate this effect. Furthermore, we set $\mtry = 5$ for all simulations. The SIR (SAVE) procedure always selects the first principal SDR direction estimated with the standard SIR (SAVE) algorithm with 10 slices. We included these benchmarks to ensure that local structure informs the LSVI algorithm. Local SIR is related to the KNN-SIR method proposed in~\citet{wang2009}. For each prediction point, we selected $k \in \set{\max(10, p), 25, 50, 100}$ of the prediction point's nearest neighbors and estimated the leading SIR direction restricted to these neighbors. We recorded the best trace correlation obtained by the four directions, which correspond to each value of $k$.

Figure \ref{fig:lsvi_results_trcor_stn_p10} displays boxplots of the trace correlations achieved by the four methods at $p = 10$ while the signal-to-noise ratio varied. For most scenarios, the DRF procedure performs best with slowly degrading performance as the signal-to-noise ratio decreases. In particular, as the signal-to-noise ratio decrease, the DRF based estimates degenerate to the performance of SAVE. This behavior is especially evident in Scenario 2, where the principal SAVE direction is highly informative. In conclusion, LSVIs produced by DRFs adequately describe the local dimension reduction structure, especially when the signal-to-noise ratio is greater than one.

We include the analysis of how LSVI estimation performed as the number of uninformative covariates increased in Section~\ref{sec:add_figs} of the Supplementary Materials. Our analysis indicated that LSVIs are relatively robust to the number of uninformative covariates, unlike Local SIR, which performed poorly due to the curse of dimensionality.

\begin{figure}[tbp]
    \centering
    \includegraphics[width=\textwidth]{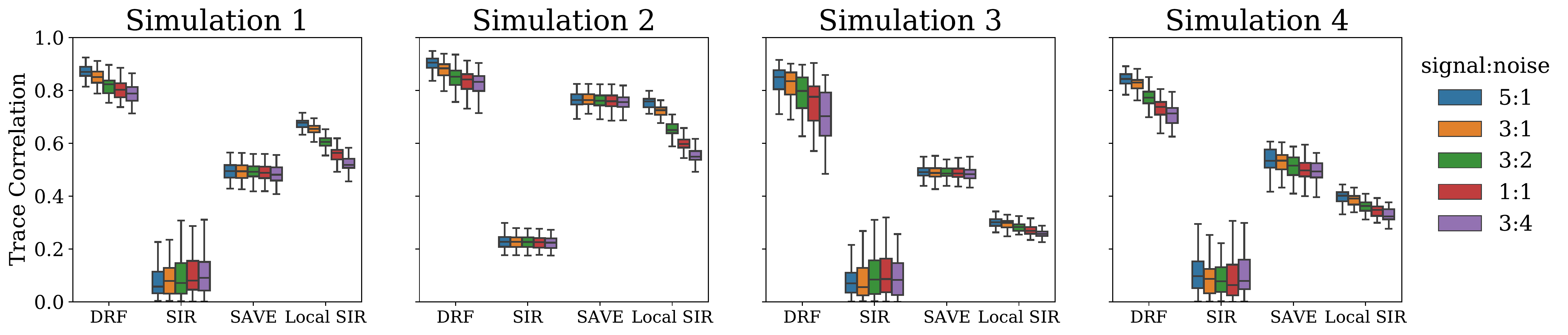}
    \caption{Boxplots of the trace correlations between the local gradients and the LSVIs estimated by the four methods outlined in Section \ref{subsec:lsvi_sim} for different signal-to-noise ratios. Higher values of the trace correlation are better.}
    \label{fig:lsvi_results_trcor_stn_p10}
\end{figure}

\section{Real Data Applications}\label{sec:real_data}

In this section, we focus on the empirical performance of dimension reduction forests on real data sets. To illustrate our algorithm's predictive value, we applied DRFs to a variety of real-world regression tasks. Furthermore, we used the LSVIs produced by a DRF to understand the monthly variation of PM2.5 concentration in Beijing, China.

\subsection{Predictive Performance on Real Data Sets}

We applied the dimension reduction forest algorithm to 12 regression data sets taken from the UCI Machine Learning Repository \citep{ucirepo} and OpenML \citep{openml}. The 12 data sets represent a variety of regression tasks on real-valued covariates. Table \ref{tab:dataset_stats} in the Supplementary Materials displays each data set's sample size and the number of features. We compared the DRF's predictive performance to the same six methods included in Section~\ref{subsec:predict}. For each data set, we performed 15 different rounds of 10-fold cross-validation. We calculated the average out-of-sample MSE over the 10-folds and recorded the percentage (100\%) improvement over the random forest baseline. As before, we recorded the best value over the range of hyperparameter settings displayed in Table \ref{tab:param_tune} of the Supplementary Materials. Table \ref{tab:real_data} contains the results for the 12 data sets.

\begin{table}[tbp]
    \scriptsize
    \resizebox{\columnwidth}{!}{
    \begin{tabular}{@{}llllllll@{}} \toprule
    & DRF & SIR + RF & SAVE + RF & NW Kernel & SIR + NW Kernel & SAVE + NW Kernel \\
    \bottomrule[\lightrulewidth]
Abalone & \cellcolor{bestcolor!25} ${5.92\ (0.42)}$ & $4.66\ (0.51)$ & $3.48\ (0.87)$ & $-18.79\ (0.49)$ & $-27.64\ (0.52)$ & $-66.74\ (1.95)$ \\
Auto Price & $-10.32\ (3.89)$ & $-58.32\ (18.19)$ & $-188.25\ (31.94)$ & $-58.89\ (10.41)$ & $-62.45\ (13.99)$ & $-210.93\ (28.86)$ \\
Bank8FM & \cellcolor{bestcolor!25} ${13.23\ (0.34)}$ & $6.13\ (0.79)$ & $6.82\ (0.31)$ & $-135.67\ (1.06)$ & $-135.85\ (1.06)$ & $-135.83\ (1.05)$ \\
Body Fat & \cellcolor{bestcolor!25} ${13.20\ (3.17)}$ & $-20.59\ (8.03)$ & $-227.61\ (48.19)$ & $-291.88\ (14.00)$ & $-291.88\ (14.00)$ & $-913.51\ (123.63)$ \\
CPU Small & $-14.67\ (1.22)$ & $-48.51\ (3.60)$ & $-36.89\ (3.20)$ & $-57.84\ (3.71)$ & $-65.71\ (7.01)$ & $-148.87\ (7.47)$ \\
Fish Catch & \cellcolor{bestcolor!25} ${17.46\ (9.86)}$ & $9.33\ (12.77)$ & $-110.71\ (40.41)$ & $-47.11\ (15.70)$ & $-746.05\ (230.58)$ & $-2249.80\ (353.29)$ \\
Kin8nm & $50.49\ (0.25)$ & $24.03\ (1.01)$ & \cellcolor{bestcolor!25} ${50.94\ (0.31)}$ & $33.04\ (0.34)$ & $32.81\ (0.34)$ & $32.96\ (0.35)$ \\
Liver & \cellcolor{bestcolor!25} ${0.30\ (0.87)}$ & $-4.10\ (1.76)$ & $-5.16\ (1.44)$ & $-13.54\ (2.28)$ & $-13.67\ (2.40)$ & $-17.16\ (2.12)$ \\
Mu284 & \cellcolor{bestcolor!25} ${20.96\ (2.22)}$ & $-8.58\ (9.69)$ & $-2.86\ (10.07)$ & $-13.78\ (3.28)$ & $-69.76\ (8.78)$ & $-165.74\ (19.20)$ \\
Puma32H & \cellcolor{bestcolor!25} ${6.07\ (0.32)}$ & $-991.25\ (20.71)$ & $-175.48\ (3.56)$ & $-1757.70\ (12.61)$ & $-1761.43\ (13.30)$ & $-1757.70\ (12.61)$ \\
Puma8NH & \cellcolor{bestcolor!25} ${2.10\ (0.25)}$ & $0.79\ (0.33)$ & $0.57\ (0.44)$ & $-33.84\ (0.69)$ & $-33.88\ (0.69)$ & $-33.87\ (0.69)$ \\
Wisconsin & \cellcolor{bestcolor!25} ${3.53\ (0.84)}$ & $-0.92\ (2.32)$ & $-3.18\ (2.36)$ & $-8.36\ (1.25)$ & $-8.36\ (1.25)$ & $-9.36\ (1.13)$ \\
    \bottomrule
    \end{tabular}}
    \caption{Percentage (100\%) improvement over the RF baseline: 1 - (MSE of method) / (MSE of random forest). The standard deviation over the 15 runs is displayed in parentheses. Methods with the best improvement (if any) over the RF are shaded in gray.}
    \label{tab:real_data}
\end{table}

Overall, the DRF performs better or equivalent to the RF baseline on most data sets. Specifically, the DRF outperforms the RF on 9 out of 12 data sets. The largest improvements occur either at small sample sizes or where SIR (SAVE) features already result in an improvement. Furthermore, the DRF also saw an improvement in cases where including SIR (SAVE) features hindered the RF's performance, e.g., on Puma32H and Wisconsin. The only times the DRF method performed worse than the RF is when the SDR features also performed poorly. We conclude that DRFs are a simple alternative to complicated SDR feature engineering with the added benefit that they often improve performance due to their ability to adapt to local structures.

\subsection{PM2.5 Concentration in Beijing, China}\label{subsec:beijing}

Finally, we demonstrate how to use the LSVIs generated by the DRF algorithm to understand the seasonal variation of PM2.5 (particulate matter) concentration ($\mu g/ m^3$) in the Chaoyang district of Beijing, China. The raw data was collected by the Urban Air project (Urban Computing Team, Microsoft Research) and consists of meteorological and air quality data collected hourly by 437 air quality stations located in 43 Chinese cities.  The data was collected over one year, from May 1st, 2014 to April 30th, 2015. We limited the analysis to observations recorded in Beijing's Chaoyang district because it is the largest and most populous district. For more details on this data set, see \citet{zheng2015}.

The analysis's goal is to quantify how different meteorological covariates influence the predictions of PM2.5 concentration throughout the year. We included five covariates in the analysis: \texttt{month}, \texttt{temperature} ($^\circ$C), (atmospheric) \texttt{pressure} (hPa), \texttt{humidity} (\%), and \texttt{wind speed} (m/s). We coded the month as an integer from 1 to 12, starting with January. Furthermore, we standardized all covariates so that they have mean zero and unit variance. The final processed data set contained $n = 25,794$ observations with $p = 5$ covariates. For reference, the pairwise scatter plots and marginal histograms of the covariates used in the analysis are displayed in Figure \ref{fig:beijing_pairplot} of the Supplementary Materials.

We fit a dimension reduction forest to this data set with $M = 500$ trees and no feature screening. We selected 3 as the minimum leaf node size using an 80\%--20\% train-test split to search over $n_{min} \in \set{3, 5, 10}$. Then, we re-fit the DRF to the full data set with $n_{min} = 3$. The model's in-sample $R^2$ is 0.92, indicating a good fit to the data.

To discern how the covariates affected the prediction of PM2.5 concentration, we plotted the LSVI loadings' marginal distributions in Figure~\ref{fig:beijing_importances}. Since the LSVI estimates are invariant to sign changes, we forced the temperature loadings' sign positive. The largest loadings are on month, humidity, and temperature. Air pressure and wind speed play a less significant role. For comparison, we included the global permutation-based importance extracted from a traditional random forest and a DRF in Figure \ref{fig:beijing_imp} of the Supplementary Materials. These results corroborate wind speed's lack of global importance.

\begin{figure}[tbp]
    \centering
    \includegraphics[width=\textwidth]{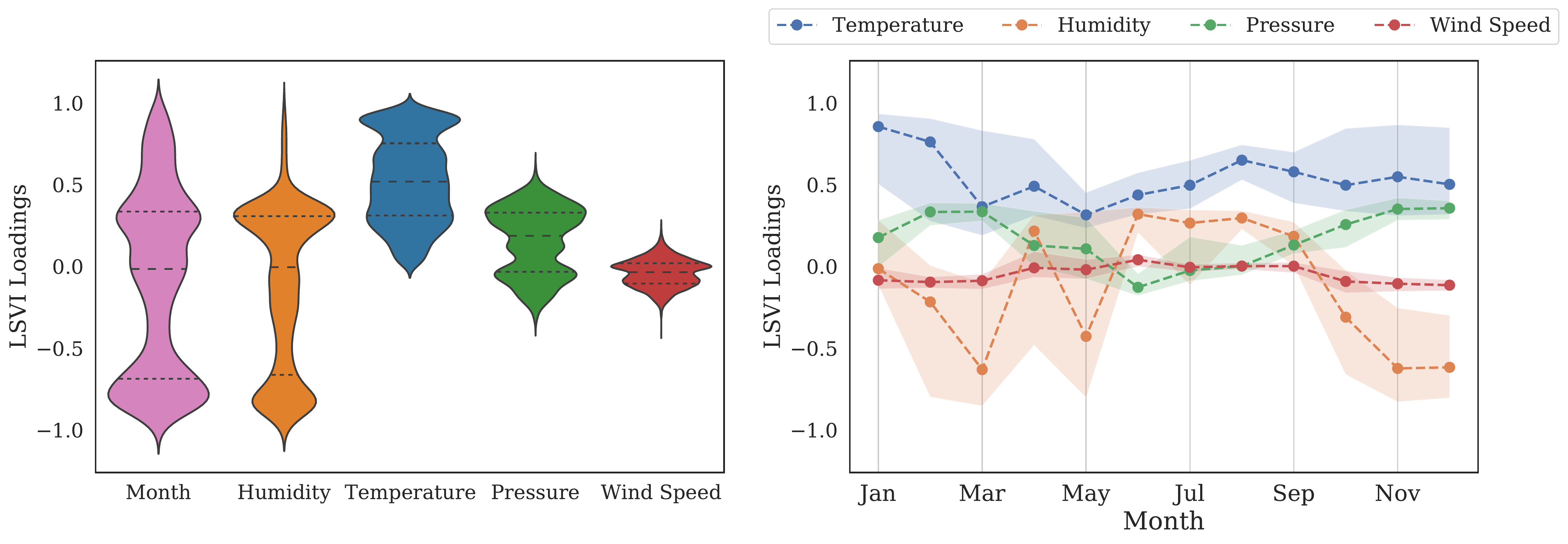}
    \caption{Violin plot of each LSVIs' marginal distribution with the temperature loading's sign forced positive (left). Dashed lines indicate the median and dotted lines are the interquartile range. The distribution of the LSVIs loadings as a function of month (right). Dotted lines indicate the median and shaded regions denote the interquartile range.}
    \label{fig:beijing_importances}
\end{figure}

To assess how these effects varied by month, we plotted the loadings' distributions as a function of month in Figure~\ref{fig:beijing_importances}. During the summer months (Jun, Jul, Aug, and Sep), the regression function primarily varies along the subspace spanned by two meteorological variables: \texttt{temperature} + \texttt{humidity}. In contrast, the regression function varies along the subspace spanned by \texttt{temperature} + \texttt{pressure} - \texttt{humidity} in the remaining months. This finding indicates that PM2.5 concentration's cause varies by season~\citep{chen2020}.

\section{Discussion}\label{sec:discussion}

This work proposed dimension reduction forests, a new nonparametric estimator that can adapt to local dimension reduction structure and measure local variable importance through a novel application of sufficient dimension reduction. We presented theory on the sufficient conditions for using SDR within the forest's splitting rule. A simulation study demonstrated that dimension reduction forests often outperform traditional random forests in predictive tasks, especially when the regression function contains a local dimension reduction structure. Also, we demonstrated the accuracy of our proposed local subspace variable importance procedure. Finally, we established our model's effectiveness on 12 real-world regression tasks and used it to study seasonal variation in PM2.5 concentration in Beijing.

The use of sufficient dimension reduction within random forests opens up an exciting line of future work on designing more efficient adaptive kernels. For example, although SIR and SAVE are robust to violations of the linearity and constant variance conditions, these assumptions can be removed entirely using computationally expensive semiparametric estimators~\citep{ma2012}. Furthermore, high-dimensional problems pose a challenge for our method, which we accounted for through an optional variable screening step. Another approach is to utilize sparse SDR estimators~\citep{lin2018, lin2019}. However, the computational complexity and sample size requirements of these two approaches are prohibitive for most practical problems when used within a dimension reduction forest. Modifications to improve the scalability of the aforementioned estimators within a DRF are an area of research interest. Further research directions include establishing the asymptotic properties of DRFs and their associated local subspace variable importance measure. Regardless, dimension reduction forests remain a flexible and interpretable nonparametric estimator. A repository containing all data sets and scripts used to run the analyses in this article is available on GitHub~\citep{drforest2021}.

%\bigskip
%\begin{center}
%{\large\bf SUPPLEMENTARY MATERIAL}
%\end{center}

%\begin{description}
%
%%\item[Supplementary Information:]  We provide an overview of sufficient dimension reduction, a computational complexity analysis of Algorithm~\ref{alg:drf}, proofs of the theoretical results, further details on the simulation study and data sets, additional simulations assessing LSVI estimation, and additional figures from the real data analysis. (pdf file)
%
%\item[Python-package for DRFs:] Python-package {\tt drforest} is available at~\url{https://github.com/joshloyal/drforest}. The package contains all data sets and scripts used to run the analyses in this article. (GitHub repository)
%
%\end{description}

\bibliographystyle{asa}
\bibliography{reference}

\newpage
\input{supplement}

\end{document}

%% file: supplement.tex
\clearpage
\baselineskip=18pt

\setcounter{page}{1}
\setcounter{section}{0}
\setcounter{equation}{0}
\setcounter{table}{0}
\setcounter{figure}{0}
\pagenumbering{arabic}
\renewcommand{\thesubsection}{S.\arabic{subsection}}
\renewcommand\theequation{S.\arabic{equation}}
\renewcommand\thetable{S.\arabic{table}}
\renewcommand\thefigure{S.\arabic{figure}}
\renewcommand{\thealgorithm}{S.\arabic{algorithm}}
\renewcommand{\thethm}{S.\arabic{thm}}
\renewcommand{\thecorollary}{S.\arabic{corollary}}

\begin{center}
{\Large\textbf{Supplementary Material for \\
``Dimension Reduction Forests: Local Variable Importance using Structured Random Forests"}} \\

\smallskip

\if0\blind
{
{\large Joshua Daniel Loyal, Ruoqing Zhu, Yifan Cui, and Xin Zhang}
}\fi
\if1\blind
{
\bigskip
}\fi
\end{center}

\subsection{Overview of Inverse Regression Methods} \label{subsec:sir_save_algo}

In this section, we detail the SIR and SAVE algorithms used to estimate the central subspace in a dimension reduction tree. Both methods fall under the framework of inverse regression and can be formulated as solving the following generalized eigenvalue problem:
\begin{equation*}
\Lambda B = \lambda \Sigma B,
\end{equation*}
where $\Sigma$ is the covariance matrix of $\bX$, the columns of $B$ span at least part of the central subspace, and $\Lambda$ is a method specific matrix that is a function of the moments of $\bX \mid Y$. Since these estimators reverse the usual dependence of $Y$ on $\bX$, they are known as inverse regression methods in the literature. SIR uses first order moment information with
\begin{equation}\label{eq:sir}
\Lambda_{\text{SIR}} = \Cov{\E{\bX|Y}}.
\end{equation}
It is easy to see that SIR fails to recover $B$ when $\E{\bX|Y} = 0$, e.g., when $\E{\bX|Y}$ is symmetric about $\E{\bX}$. A solution proposed in SAVE is to utilize second moment information. In particular, SAVE uses the matrix
\begin{equation}\label{eq:save}
\Lambda_{\text{SAVE}} = \E{(I_p - \Var{\bX | Y})^2},
\end{equation}
where $I_p$ is the $p$-dimensional identity matrix. Algorithm \ref{alg:sir} outlines the SIR algorithm, and Algorithm \ref{alg:save} details the SAVE algorithm. Note that the linearity condition is sufficient for SIR to recover a direction in $\Sxy$, while SAVE requires both the linearity condition and the constant variance condition.

\begin{algorithm}[htbp]
\caption{Sliced Inverse Regression}\label{alg:sir}
\begin{algorithmic}[1]
\Procedure{SolveSIR}{data set $\dataset = \{\bX_i, Y_i\}_{i=1}^n$}
\State $U, \hat{\mu} \leftarrow $ \textproc{Center}($X$)
\State $Q, R \leftarrow$ \textproc{QRDecomposition}($U$)
\State $Z \leftarrow \sqrt{n} \ Q$
\Comment Whiten the data.
\For{slice $S_k$ in \textproc{GetSlices}($Z$, \texttt{n\_slices})}
\State $\hat{\mathbf{\mu}}_{S_k} \leftarrow \frac{1}{|S_k|}\sum_{i=1}^n \mathbf{Z}_i \ind{\mathbf{Z}_i \in S_k}$
\EndFor
\State $\hat{\Lambda}_{\text{SIR}} \leftarrow \sum_{k=1}^{\texttt{n\_slices}} \frac{|S_k|}{n} (\hat{\mu}_{S_k} - \hat{\mu})(\hat{\mu}_{S_k} - \hat{\mu})^T$
\Comment See Equation (\ref{eq:sir}).
\State Solve $\hat{\Lambda}_{\text{SIR}} \hat{\Gamma} = \hat{\lambda} \hat{\Gamma}$
\State $\hat{B} \leftarrow (\sqrt{n} R)^{-1}\hat{\Gamma}$
\Comment Solved via back-substitution.
\State \textbf{output} Estimated directions $\hat{B}$ and eigenvalues $\hat{\lambda}$.
\EndProcedure
\end{algorithmic}
\textproc{Center} calculates the empirical feature-wise means $\hat{\mu}$ and centers the data: $\mathbf{U}_i = \bX_i - \hat{\mu}$. \textproc{GetSlices} splits the rows of $Z$ into \texttt{n\_slices} contiguous segments with roughly equal sample-size. Note that $Z$ is assumed to be presorted in terms of the target $Y$.
\end{algorithm}

\begin{algorithm}[htbp]
\caption{Sliced Average Varience Estimation}\label{alg:save}
\begin{algorithmic}[1]
\Procedure{SolveSAVE}{data set $\dataset = \{\bX_i, Y_i\}_{i=1}^n$}
\State $U, \hat{\mu} \leftarrow $ \textproc{Center}($X$)
\State $Q, R \leftarrow$ \textproc{QRDecomposition}($U$)
\State $Z \leftarrow \sqrt{n} \ Q$
\Comment Whiten the data.
\State $\hat{\Lambda}_{\text{SAVE}} \leftarrow 0_{p \times p}$
\For{slice $S_k$ in \textproc{GetSlices}($Z$, \texttt{n\_slices})}
\State $\tilde{Z}_{S_k} \leftarrow $ \textproc{CenterSlice}($Z, S_k$)
\State $\hat{\Lambda}_{\text{SAVE}} \leftarrow \hat{\Lambda}_{\text{SAVE}} + \frac{|S_k|}{n} (I_{|S_k|} - \tilde{Z}_{S_k}^T \tilde{Z}_{S_k})^2$
\Comment See Equation (\ref{eq:save}).
\EndFor
\State Solve $\hat{\Lambda}_{\text{SAVE}} \hat{\Gamma} = \hat{\lambda} \hat{\Gamma}$
\State $\hat{B} \leftarrow (\sqrt{n} R)^{-1}\hat{\Gamma}$
\Comment Solved via back-substitution.
\State \textbf{output} Estimated directions $\hat{B}$ and eigenvalues $\hat{\lambda}$.
\EndProcedure
\end{algorithmic}
\textproc{Center} calculates the empirical feature-wise means $\hat{\mu}$ and centers the data: $\mathbf{U}_i = \bX_i - \hat{\mu}$. \textproc{CenterSlice} is the same as \textproc{Center} but restricted to data in a given slice $S_k$. It returns the centered data for slice $S_k$. \textproc{GetSlices} splits the rows of $Z$ into \texttt{n\_slices} contiguous segments with roughly equal sample-size. Note that $Z$ is assumed to be presorted in terms of the target $Y$.
\end{algorithm}

\clearpage

\subsection{Computational Complexity Analysis}\label{comp_complexity}

An advantage of random forests over other machine learning methods is their relatively fast training speed, so transferring that speed to dimension reduction forests is paramount. We demonstrate that DRF's theoretical computational complexity is comparable to a standard random forest when $p < n$. We only consider DRFs built without the optional screening step in step (a) of Algorithm \ref{alg:drf}; however, we allow the random selection of $\mtry \leq p$ variables at each node. The following theorem quantifies the computational complexity of building a single dimension reduction tree.

\begin{thm}\label{thm:drt_complexity}
Let $C(s)$ denote the computational complexity of forming a single node on a data set with sample size $s$. Assume that the trees are grown such that all possible sample sizes assigned to a child node are equally probable, i.e., uniformly distributed on $\set{1, \dots, n_A}$ where $n_A$ is the number of samples in node $A$. Also, assume $\mtry$ variables are used to estimate each splitting direction. If there exist constants $c_1, c_2 > 0$ such that $c_1 (s\mtry^2 + s \log(s)) \leq C(s) \leq c_2 (s\mtry^2 + s \log(s))$ (for all $s \geq 1$), then the computational complexity for building a dimension reduction tree on a data set with sample size $n$ is $\bigoh(\mtry^2 n \log(n) + \mtry n\log^2(n))$.
\end{thm}

\begin{proof}
The proof is a modification of the proof of Theorem 5.6 in \citetSup{louppe2014sup}. Let $T(n)$ denote the computational complexity for building a decision tree from a data set $\dataset$ with $n$ samples. For simplicity of the proof, we assume that the decision trees can be fully grown such that each leaf node contains a single sample. In this framework, the decision tree building process is characterized by the following recurrence relation:
\begin{equation*}
\begin{cases}
T(1) = c_1, \\
T(n) = C(n) + T(n_L) + T(n_R),
\end{cases}
\end{equation*}
where $c_1$ is the time complexity to make a leaf node, and $n_L$ and $n_R$ are the number of samples in the left and right child node. The average time complexity assumes $n_L$ and $n_R$ are uniformly distributed on $\{1, \dots, n-1\}$. Thus the recurrence relation in the average case is
\begin{equation*}
\begin{cases}
T(1) = c_1, \\
T(n) = C(n) + \frac{1}{n-1}\sum_{i=1}^{n-1}(T(i) + T(n - i)).
\end{cases}
\end{equation*}
The proof proceeds in standard fashion. Let $C(n) = c_3 (n\mtry^2 + n\log(n))$. By symmetry and multiplying by $(n-1)$ the recurrence becomes
\begin{equation}\label{eq:n1}
(n-1)T(n) = (n-1)c_3 (n \mtry^2 + n\log(n)) + 2 \sum_{i=1}^{n-1}T(i).
\end{equation}
In addition, for $n \geq 3$, substituting $n$ with $n-1$ yields
\begin{equation}\label{eq:n2}
(n-2)T(n-1) = (n-2)c_3((n-1) \mtry^2 + (n-1)\log(n-1)) + 2 \sum_{i=1}^{n-2} T(i).
\end{equation}
Define $S(n) = \frac{T(n)}{n}$. Now subtracting Equation (\ref{eq:n2}) from Equation (\ref{eq:n1}) and dividing by $n(n-1)$ we have
\begin{align*}
S(n) &= S(n-1) + 2 c_3 \frac{\mtry^2}{n} + c_3 \frac{2}{n}\log(n-1) + c_3 \log(\frac{n}{n-1}), \\
    &= c_1' + c_3 \sum_{i=2}^n \frac{2}{i}(\mtry^2 + \log(i-1)) + \log(\frac{i}{i-1}), \\
    &= c_1' + c_3 \log(n) + c_3 \sum_{i=1}^n \frac{2}{i}(\mtry^2 + \log(i-1)), \\
    &\leq c_1' + c_3 \log(n) + 2 c_3 (\mtry^2 + \log(n)) \sum_{i=2}^n \frac{1}{i}, \\
    &= c_1' + c_3 \log(n) + 2 c_3 (\mtry^2 + \log(n)) (H_n - 1), \\
    &= \bigoh(H_n\mtry^2 + H_n \log(n)),
\end{align*}
where $c_1' = c_1 / n$ and $H_n$ is the $n$th harmonic number. To finish the bound we note that $H_n \sim \log(n)$. As a result we see that the average case computational complexity of building this decision tree is
\begin{equation*}
T(n) = \bigoh(\mtry^2 n \log(n) + \mtry n \log^2(n)),
\end{equation*}
where we upper bounded the last term with a multiple of $\mtry$ in order to relate it to the normal decision tree bound.

\end{proof}

As we shortly show, the DRT splitting procedure satisfies the bounds in Theorem~\ref{thm:drt_complexity}. Furthermore, since the construction of dimension reduction forests is linear in the number of trees $M$, we immediately have the following corollary to Theorem \ref{thm:drt_complexity}.

\begin{corollary}\label{cor:drf_complexity}
The computational complexity of the dimension reduction forest algorithm is $\bigoh(M m_{try} n \log(n) (m_{try}+ \log(n)))$.
\end{corollary}

\begin{proof}
To apply Theorem \ref{thm:drt_complexity} to our dimension reduction tree algorithm, we must demonstrate that the computational complexity of our splitting procedure satisfies the bounds listed in the theorem. For simplicity, we assume that the dimension reduction tree is built using a SIR (SAVE) splitting rule all the way down the tree and without the variable screening step. In this case, the computation time per node is due to the calculation of the SIR and SAVE directions and the search for the best split. SIR and SAVE both require a QR decomposition, an eigendecomposition of the $\Lambda$ matrix, and sorting the data set with respect to the response. These operations have a combined complexity of
\begin{equation*}
\bigoh(s m_{try}^2 + m_{try}^3 + s \log(s)) = \bigoh(s m_{try}^2 + s \log(s))
\end{equation*}
according to \citetSup{golub2012sup}, where $s$ is the sample size of the node. We used the fact that the first term dominates under the assumption that $m_{try} < s$, which we require due to the limitations of SIR (SAVE). Once these directions are computed, a DRT searches for the best split along the leading direction. This operation is bounded by the theoretical complexity of sorting, i.e., $\bigoh(s \log(s))$. Putting this all together we have that
\begin{equation}\label{eq:drf_comp}
C(s) = \bigoh(s m_{try}^2 + s \log(s))
\end{equation}
for a dimension reduction tree.

Equation (\ref{eq:drf_comp}) allows us to apply Theorem \ref{thm:drt_complexity} to determine the computational complexity of a single DRT. The full dimension reduction forest algorithm is linear in the number of tress $M$, so that the computational complexity of the DRF algorithm is
\begin{equation*}
\bigoh(M m_{try} n \log(n) (m_{try}+ \log(n))).
\end{equation*}
\end{proof}

The time complexity of a standard random forest under similar assumptions~\citeSup{louppe2014sup} is $\bigoh(M \mtry n \log(n)^2)$. If we choose $m_{try} \leq \log(n)$, then the time complexity of a DRF is bounded by the standard random forest's complexity. Note that the worst-case performance is when $m_{try} = p$. In this case, the DRF's performance is worse than a standard RF by a factor of $p / \log(n)$, which is small for many applications where $p << n$.

\subsection{Proofs of Theorem \ref{thm:sir} and Theorem \ref{thm:save}}\label{subsec:sir_save_proofs}

Theorem \ref{thm:sir} verifies that recursive application of SIR within a dimension reduction tree continues to estimate a part of the local central subspace. First for some preliminaries. One can show that Assumption \ref{asmp:lcm} implies that
\[
\E{\bX - \E{\bX | \bX \in A}|B^{\rm T}_A\bX, \bX \in A} = P_{B_A}(\Sigma_A)^{\rm T}(\bX - \E{\bX | \bX \in A}),
\]
where $P_{B_A}(\Sigma_A) = B_A(B_A^{\rm T}\Sigma_A B_A)^{-1}B^{\rm T}_A \Sigma_A$ is the projection matrix onto the central subspace $\SxAy$ (Lemma 1.1 in \citetSup{bing2018sup}). Also, if Assumption \ref{asmp:lcm} and Assumption \ref{asmp:ccv} both hold, then one can show that $\Var{\bX | B_A^{\rm T}\bX, \bX \in A} = \Sigma_A Q$, where $Q = I_{p} - P_{B_A}(\Sigma_A)$ (Corollary 5.1 in \citetSup{bing2018sup}). We are now ready to prove Theorem~\ref{thm:sir} and Theorem~\ref{thm:save}.

\begin{proof}[Proof of Theorem \ref{thm:sir}]
Without loss of ambiguity, all expectations in this proof are conditioned on $\bX$ being contained in the parent node $A$. The proof is little more than an application of the law of total expectation. We only carry out the proof for the left child node $A_L$. The proof for the right child node $A_R$ is exactly the same. Without loss of generality, assume $\E{\bX \mid \bX \in A} = 0$. We have that
\begin{align*}
\E{\bX | Y, \bX \in A_L} &= \E{\bX|Y, \beta^{\rm T}\bX \leq c}, \\
                       &=\E{\E{\bX|Y, \beta^{\rm T}\bX \leq c, B_A^{\rm T}\bX}|Y, \beta^{\rm T}\bX \leq c}, \\
                       &= \E{\E{\bX|B_A^{\rm T}\bX} \lvert Y, \beta^{\rm T} \bX \leq c}, \\
                       &= P_{B_A}(\Sigma_A)^{\rm T} \E{\bX|Y, \beta^{\rm T} \bX \leq c}, \\
                       &= \Sigma_A P_{B_A}(\Sigma_A) \Sigma^{-1}_A\E{\bX|Y, \beta^{\rm T} \bX \leq c}, \\
                       &= \Sigma_A P_{B_A}(\Sigma_A) \Sigma^{-1}_A\E{\bX|Y, \bX \in A_L}.
\end{align*}
Multiplying both sides by $\Sigma^{-1}_A$ completes the proof. The third line is due to the fact that $Y \perp \bX \mid B_A^T\bX$ (the postulated model within the node) and that $\beta^T\bX \perp \bX \mid B_A^T \bX$ since $\beta \in \spn(B_A)$. We also used the easily verifiable identity $P_{B_A}(\Sigma_A)^T = \Sigma^{-1}_AP_{B_A}(\Sigma_A)\Sigma^{-1}_A$.
\end{proof}

Theorem \ref{thm:save} is the corresponding theorem for SAVE. It says that recursive application of SAVE within a dimension reduction tree continues to estimate a part of the local central subspace.

\begin{proof}[Proof of Theorem \ref{thm:save}]
Without loss of ambiguity, all expectations in this proof are conditioned on $\bX$ being contained in the parent node $A$. Once again, we only carry out the proof for the left child node $A_L$. The proof for the right child node $A_R$ is exactly the same. Without loss of generality, we assume $\E{\bX | \bX \in A} = 0$. Similar to the proof of Theorem \ref{thm:sir}, we just apply the law of total variation:
\begin{align*}
\Var{\bX | Y, \bX \in A_L} &= \Var{\bX | Y, \beta^T\bX \leq c}, \\
                           &= \E{\Var{\bX|Y, \beta^{\rm T}\bX \leq c, B^{\rm T}_A\bX} | Y, \beta^{\rm T}\bX \leq c} \ + \\
                             & \qquad \Var{\E{\bX | Y, \beta^{\rm T}\bX \leq c, B^{\rm T}_A \bX} | Y, \beta^{\rm T}\bX \leq c}, \\
                             &= \E{\Var{\bX|B^{\rm T}_A\bX} | Y, \beta^{\rm T}\bX \leq c} +
                                \Var{\E{\bX|B^{\rm T}_A\bX} | Y, \beta^{\rm T}\bX \leq c}, \\
                             &= \Sigma_A \left(I_{p} - P_{B_A}(\Sigma_A)\right) + P_{B_A}(\Sigma_A)^{\rm T} \Var{\bX|Y, \beta^{\rm T} \bX \leq c} P_{B_A}(\Sigma_A), \\
                             &= \Sigma_A \left(I_{p} - P_{B_A}(\Sigma_A)\right) + P_{B_A}(\Sigma_A)^{\rm T} \Var{\bX|Y, \bX \in A_L} P_{B_A}(\Sigma_A).
\end{align*}
Subtracting $\Sigma_A$ from both sides and applying the identity $\Sigma_A P_{B_A}(\Sigma_A) = P_{B_A}(\Sigma_A)^{\rm T} \Sigma_A P_{B_A}(\Sigma_A)$, we have
\begin{equation*}
\Var{\bX | Y, \bX \in A_L} - \Sigma_A = \Sigma_A P_{B_A}(\Sigma_A) \Sigma^{-1}_A \left(\Var{\bX | Y, \bX \in A_L} - \Sigma_A\right) P_{B_A}(\Sigma_A),
\end{equation*}
where the right hand side is in $\Sigma_A \spn(B_A)$ as desired.
\end{proof}

\subsection{Simulation Study and Data Set Details}\label{simulation_details}

In this section, we elaborate on various details of the simulation studies and real data analysis presented in the main text. Table~\ref{tab:param_tune} contains the hyperparameter settings used in the simulation studies of Section~\ref{sec:simulation}. Also, the three Friedman functions used to test DRFs predictive performance are
\begin{itemize}
\item \textit{Friedman 1:} $Y = 10 \sin(\pi X_1 X_2) + 20 (X_3 - 0.5)^2 + 10 X_4 + 5 X_5 + \varepsilon$,
where $\bX_i \iidsim U[0, 1]^{10}$ and $\varepsilon_i \iidsim N(0, 1)$.

\item \textit{Friedman 2:} $Y = \left(X_1^2 + \left[X_2 X_3 - \frac{1}{X_2 X_4}\right]^{2}\right)^{1/2} + \varepsilon$,
where $\bX_i$ are uniform over the hyper-rectangle $[0, 100] \times [20, 280] \times [0, 1] \times [1, 11]$ and $\varepsilon_i \iidsim N(0, 1)$.

\item \textit{Friedman 3:} $Y = \tan^{-1}\left(\frac{X_2 X_3 - \frac{1}{X_2 X_4}}{X_1}\right) + \varepsilon$,
where $\bX_i$ are uniform over the hyper-rectangle $[0, 100] \times [20, 280] \times [0, 1] \times [1, 11]$ and $\varepsilon_i \iidsim N(0, 1)$.
\end{itemize}

\begin{table}[htbp]
\small
\begin{center}
\begin{tabularx}{\textwidth}{@{}lX@{}}\toprule\toprule
NW Kernel & We use a Gaussian kernel as the kernel function. The data is standardized along each dimension. The bandwidth is then set to the recommendation in \citetSup{silverman1986sup}: $(4 / (p + 2))^{1 / (p + 4)} \, n^{-1/(p+4)}$.  \\ \\
SIR + NW Kernel & Estimates a global dimension reduction subspace using SIR with $d = p$ and the number of slices set to 10. Fits an NW Kernel to the projection of the data on to this subspace with the same settings as NW Kernel.   \\ \\
SAVE + NW Kernel & Estimates a global dimension reduction subspace using SAVE with $d = p$ and the number of slices set to 10. Fits an NW Kernel to the projection of the data on to this subspace with the same settings as NW Kernel.  \\ \\
RF & A total of 12 parameter settings. $M = 500$, $\mtry = 2, 4, 6, 1/3, \sqrt{p}, p$ and $n_{min} = 1, 5$.  \\ \\
SIR + RF& Estimates a global dimension reduction subspace using SIR with $d = p$ and the number of slices set to 10. Fits an RF model to the projection of the data onto this subspace with the same settings as RF. \\ \\
SAVE + RF & Estimates a global dimension reduction subspace using SAVE with $d = p$ and the number of slices set to 10. Fits an RF model to the projection of the data onto this subspace with the same settings as RF. \\ \\

DRF & A total of 12 parameter settings. $M = 500$,  $\mtry = 2, 4, 6, 1/3, \sqrt{p}, p$ and $n_{min} = 1, 5$. The number of slices used to estimate SIR and SAVE is set to 10.\\
\bottomrule\bottomrule
\end{tabularx}
\end{center}
\caption{Hyperparameter settings.}
\label{tab:param_tune}
\end{table}

Lastly, Table~\ref{tab:dataset_stats} displays the sample size and number of covariates of each data set used in the real data analysis of Section~\ref{sec:real_data}.

\begin{table}[htbp]
    \small
    \begin{center}
    \begin{tabular}{@{}lll@{}} \toprule
    Data Set Name & Number of Samples & Number of Features  \\
    \bottomrule[\lightrulewidth]\toprule
    Abalone & 4177 & 8 \\
    Body Fat & 245 & 14 \\
    CPU Small & 8192 & 12 \\
    Fish Catch & 158 & 7 \\
    Kin8nm & 8192 & 8 \\
    Auto Price & 159 & 15 \\
    Liver & 345 & 5 \\
    Mu284 & 284 & 9 \\
    Puma32H & 8192 & 32 \\
    Puma8NH & 8192 & 8 \\
    Wisconsin & 194 & 32 \\
    Bank8FM & 8192 & 8 \\
    \bottomrule\bottomrule
    \end{tabular}
    \end{center}
    \caption{Metadata for data sets analyzed in the real data analysis.}
    \label{tab:dataset_stats}
\end{table}

\subsection{Additional Results and Figures}\label{sec:add_figs}

In this section, we include the results of a simulation study designed to assess LSVI estimation under an increasing number of uninformative covariates and additional figures from the analysis of air pollution in Beijing, China.

\subsubsection{Additional Simulation Studies}

Here, we include a simulation study designed to measure how LSVI estimation performed under an increasing number of uninformative covariates. We used the same setup as described in Section~\ref{subsec:lsvi_sim} of the main text; however, here we set fixed $\sigma^2$ so that the signal-to-noise ratio was 3:1 and varied $p = 2, 5, 10, 15$, and 20.

Figure \ref{fig:lsvi_results_trcor} displays the boxplots of the trace correlations achieved by the four methods at $p = 2, 5, 10, 15, 20$, respectively. For all but the $p = 2$ cases, the DRF method performs best. As expected, Local SIR performs best when only informative covariates are present; although, the DRF estimates are still competitive. Furthermore, as $p$ increases, Local SIR quickly drops to the worst-performing method due to the curse of dimensionality. In contrast, the DRF estimates remain accurate.

\begin{figure}[tbp]
    \centering
    \includegraphics[width=\textwidth]{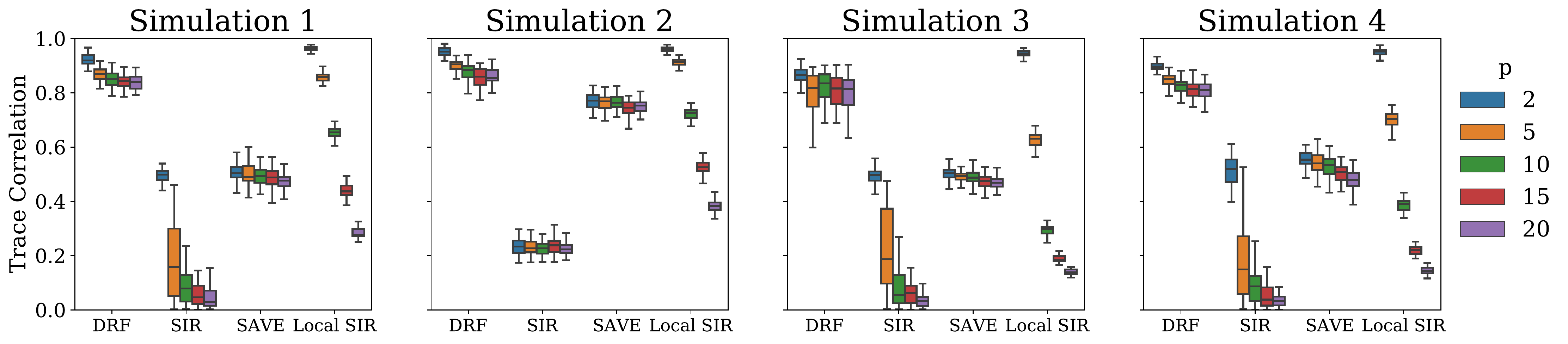}
    \caption{Boxplots of the trace correlations between the local gradients and the LSVIs estimated by the four methods outlined in Section \ref{subsec:lsvi_sim} for different numbers of covariates. The true regression functions only utilized the first two covariates. Higher values of the trace correlation are better.}
    \label{fig:lsvi_results_trcor}
\end{figure}

\subsubsection{Additional Figures}
Figure~\ref{fig:beijing_pairplot} displays the pairwise scatter plots of the meteorological features used in the analysis in Section~\ref{subsec:beijing} of the main text. The points are color coded by the value of the response variable, PM2.5 concentration. Figure~\ref{fig:beijing_imp} contains the global permutation-based variable importance calculated using a standard random forest and a dimension reduction forest trained on the data set.

\begin{figure}[htbp]
\centering
\includegraphics[width=0.6\textwidth]{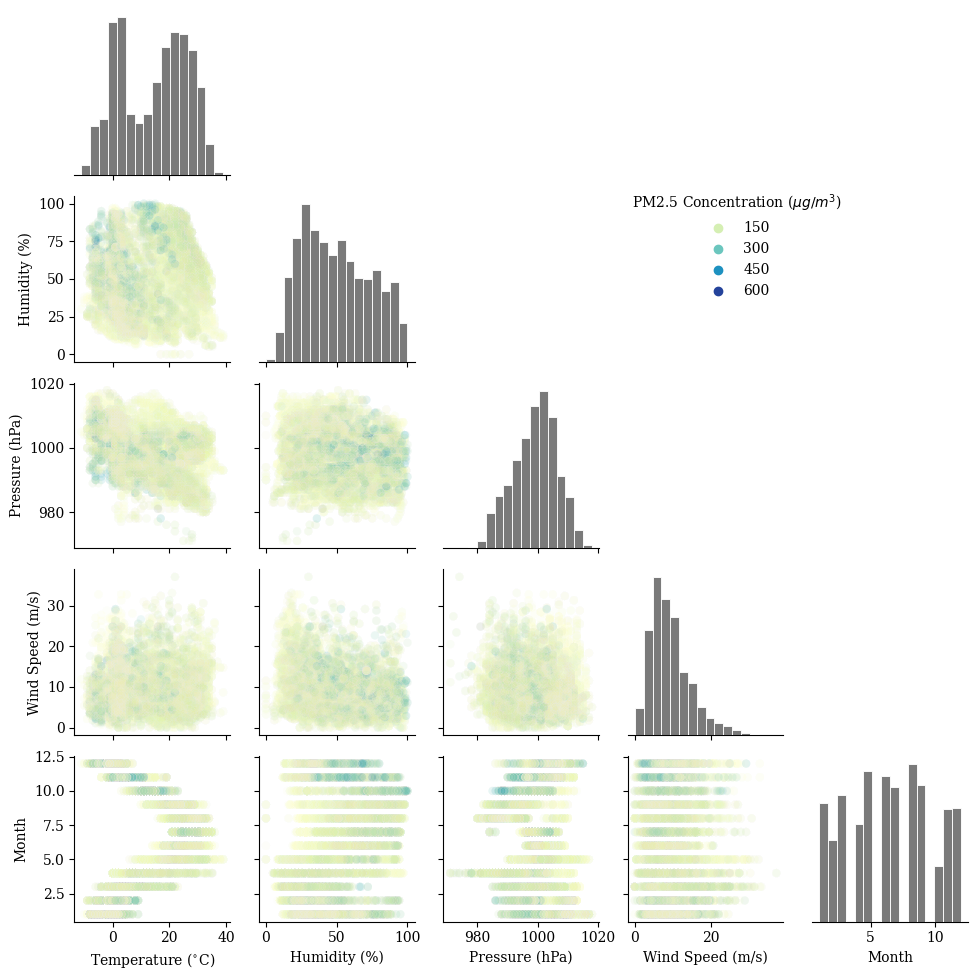}
\caption{Pairwise scatter plots of meteorological variables.}
\label{fig:beijing_pairplot}
\end{figure}

\begin{figure}[htbp]
\centering
\includegraphics[width=\textwidth]{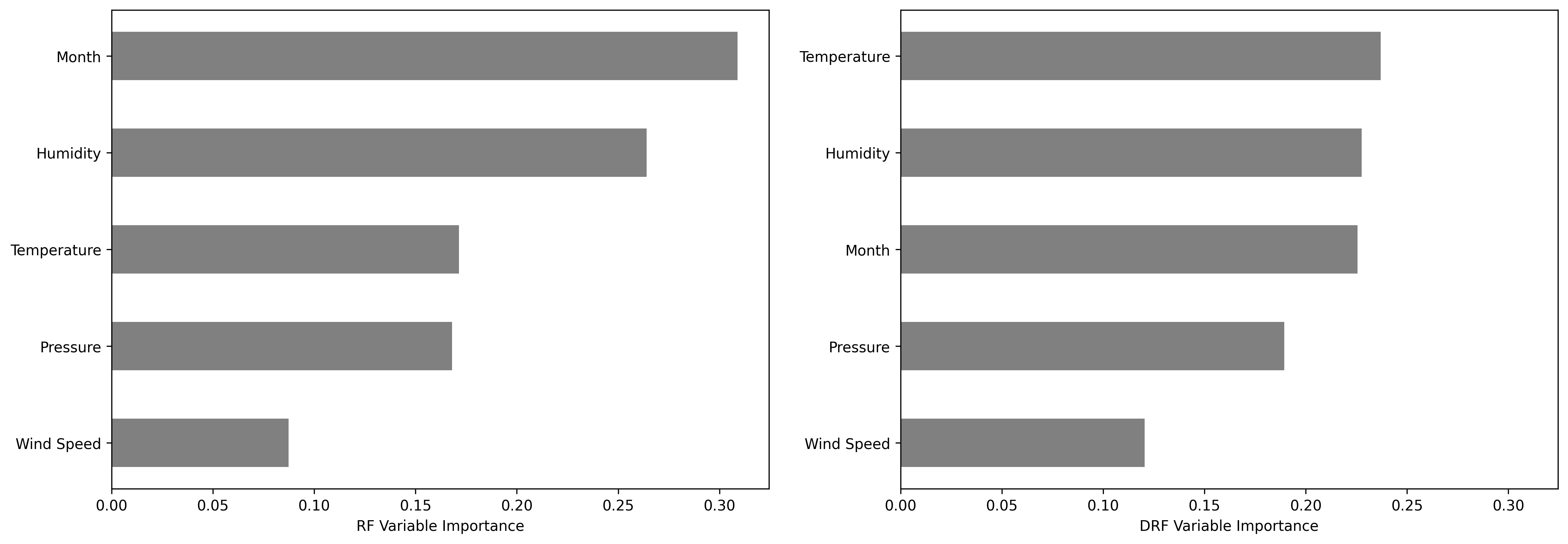}
\caption{Global permutation-based variable importance from a traditional random forest (left) and a dimension reduction forest (right) trained on the PM2.5 concentration data set.}
\label{fig:beijing_imp}
\end{figure}

\bibliographystyleSup{asa}
\bibliographySup{reference}